\documentclass[11pt]{article}
\usepackage{titling}
\usepackage{amsmath,amssymb,amsthm}
\usepackage{graphicx}
\usepackage{subcaption}
\usepackage{url}
\usepackage[noblocks]{authblk}
\newcommand{\xvec}{\boldsymbol{x}}
\newcommand{\set}[1]{\{#1\}}
\newcommand{\ins}{\text{in}}
\newcommand{\out}{\text{out}}
\newtheorem{lemma}{Lemma}
\newtheorem{theorem}{Theorem}
\newtheorem{fact}{Fact}
\newtheorem{definition}{Definition}
\newcommand{\problem}{problem}
\title{On the computational tractability of a geographic clustering
  problem arising in redistricting}
\author{Vincent Cohen-Addad}
\affil{CNRS and Sorbonne Universit\'e, Paris}
\author{Philip N. Klein\thanks{Supported by National Science Foundation grant CCF-1841954}}
\affil{Brown University}
\author{D\'aniel Marx}
\affil{Max Planck Institut f\"ur Informatik}
\begin{document}
\begin{titlingpage}
\maketitle
\begin{abstract}
  \emph{Redistricting} is the problem of dividing up a state into a
  given number $k$ of 
  regions (called \emph{districts}) where the voters in each district
  are to elect a representative. The three primary criteria are: that
  each district be connected, that the populations of the districts be
  equal (or nearly equal), and that the districts are ``compact''.
  There are multiple competing definitions of compactness, usually
  minimizing some quantity.

  One measure that has been recently promoted by Duchin and others
  (see e.g.~\cite{deford2019recombination}) is number of \emph{cut edges}.  In
  redistricting, one is generally given atomic regions out of which
  each district must be built (e.g., in the U.S., census blocks).  The
  populations of the atomic regions are given.  Consider the graph
  with one vertex per atomic region and an edge between atomic regions
  that share a boundary.  Define the weight of a vertex to be the
  population of the corresponding region.  A districting plan is a
  partition of vertices into $k$ pieces so that the parts have nearly
  equal weights and each part is connected.  The districts are
  considered compact to the extent that the plan minimizes the number
  of edges crossing between different parts.

  There are two natural computational problems: find the most compact
  districting plan, and sample districting plans (possibly under a
  compactness constraint) uniformly at random.

  Both problems are NP-hard so we consider restricting the input graph
  to have branchwidth at most $w$.  (A planar graph's branchwidth is
  bounded, for example, by its diameter.)  If both $k$ and $w$ are
 bounded by constants, the problems are solvable in polynomial time.
 For simplicity of notation, assume that each vertex has unit weight.
 We would like algorithms whose running times are of the form
 $O(f(k,w) n^c)$ for some constant $c$ independent of $k$ and $w$ (in
 which case the problems are said to be \emph{fixed-parameter
<   tractable} with respect to those parameters), we show that, under
 standard complexity-theoretic assumptions, no such algorithms exist.
 However, we do show that there exist algorithms with running
 time  $O(c^wn^{k+1})$.  Thus if the diameter of the graph is
 moderately small and the number of districts is very small, our
 algorithm is useable.
\end{abstract}
\end{titlingpage}
\section{Introduction}

For an undirected planar graph $G$ with vertex-weights and a positive integer
$k$, a \emph{connected partition}
of the vertices of $G$ is a partition into parts each of which
induces a connected subgraph.  If $G$ is equipped with nonnegative
integral vertex weights
and $[L, U)$ is an interval
we say such a partition has \emph{part-weight} in $[L, U)$ if the sum of
weights of each part lies in the interval.  If $G$ is equipped with
nonnegative edge costs, we say the \emph{cost} of such a partition is
the sum of costs of edges $uv$ where $u$ and $v$ lie in different
parts.

  \begin{figure}
\centering
  \includegraphics[scale=.5]{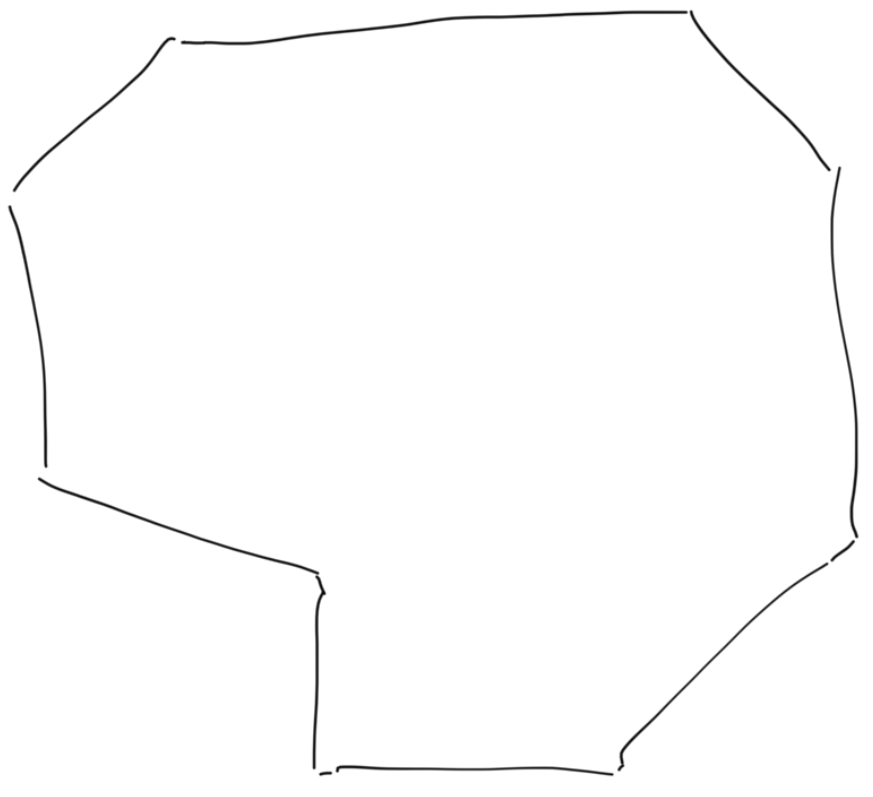}
  \caption{Outline of an (imaginary) state}
  \label{fig:outline}
\end{figure}
\begin{figure}
  \centering
  \includegraphics[scale=.5]{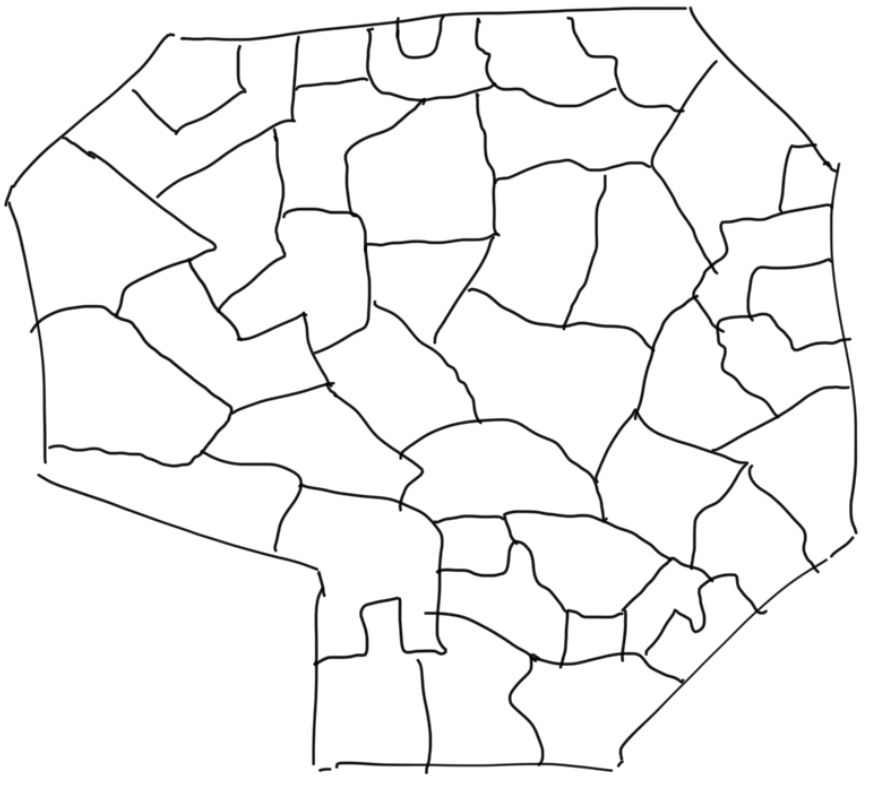}
\caption{The state subdivided into smaller regions (\emph{atoms}), e.g. census
  tracts}
\label{fig:subdivided}
\end{figure}
\begin{figure}
  \centering
  \includegraphics[scale=.5]{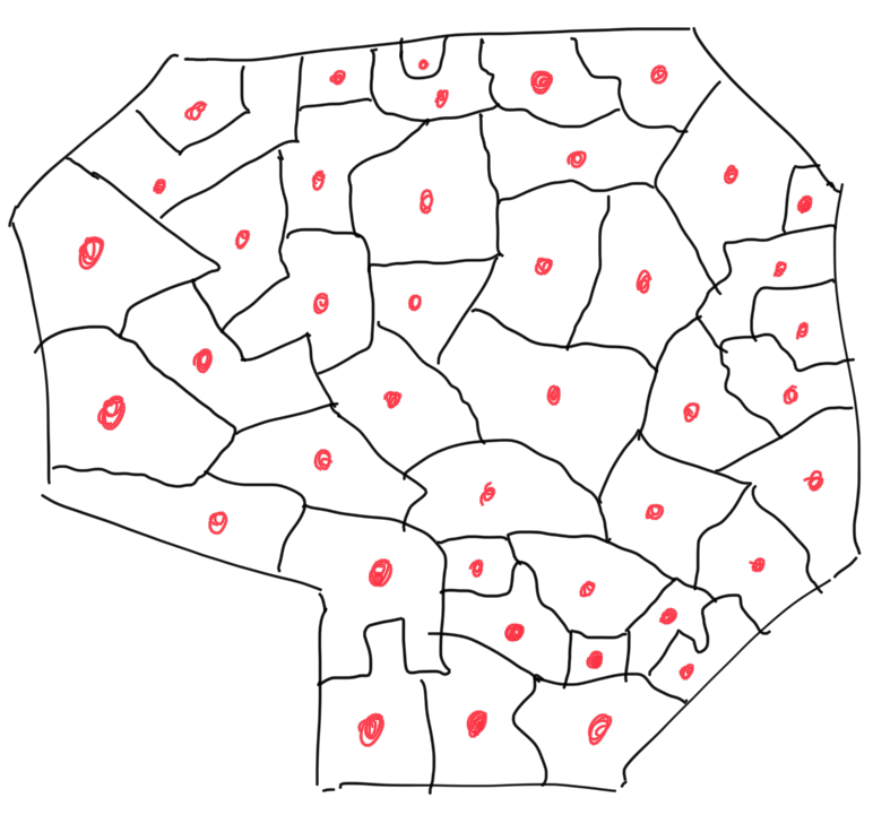}
\caption{In the planar dual, each atomic region is represented by a
  node.  (There is also a node for the single infinite region outside
  the state boundary but here we ignore that node.}
\label{fig:dual-vertices}
\end{figure}
\begin{figure}
  \centering
  \includegraphics[scale=.5]{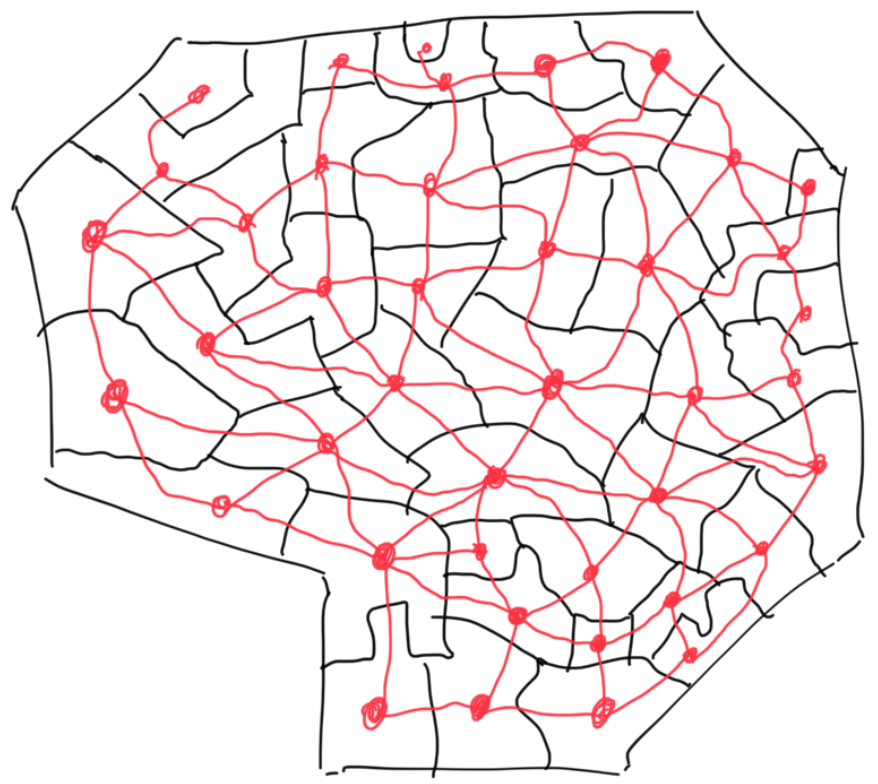}
\caption{For each maximal contiguous
  boundary segment between a pair of atomic regions, the planar dual
  has an edge between the corresponding pair of nodes.  (Again, the
  dual also has edges corresponding to segments of the boundary of the
  state but we ignore those here.}
\label{fig:dual-edges}
\end{figure}
\begin{figure}
  \centering
  \includegraphics[scale=.5]{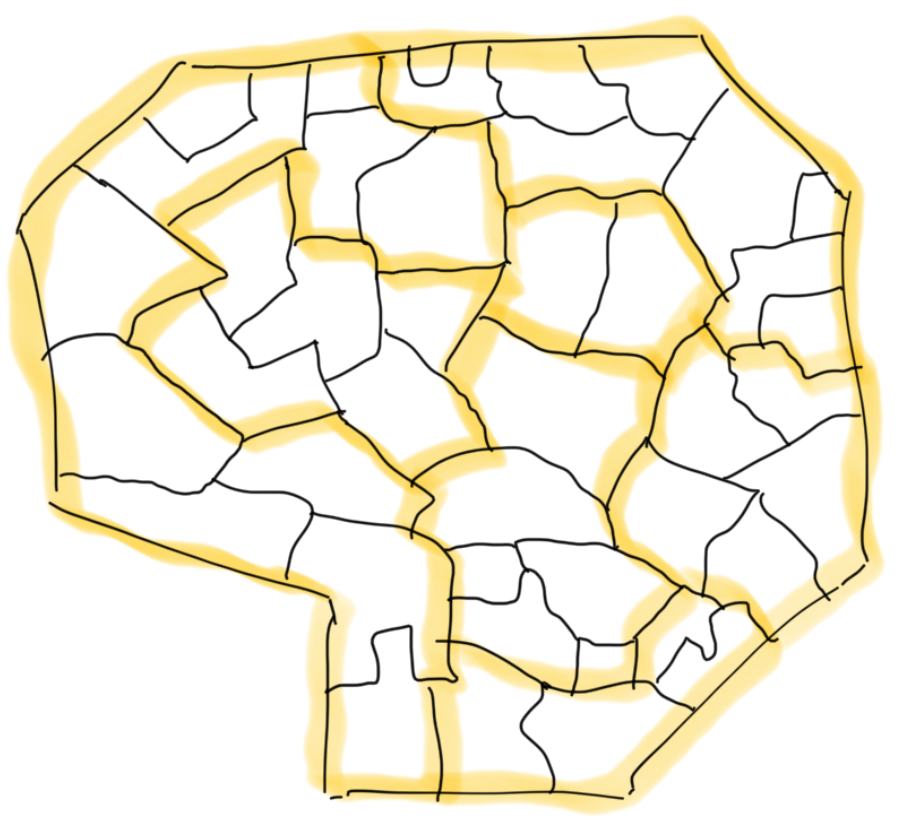}
\caption{This figure shows an example of a districting plan with seven
  districts.  Each district is the union of several atomic regions.}
\end{figure}
\begin{figure}
  \centering
  \includegraphics[scale=.5]{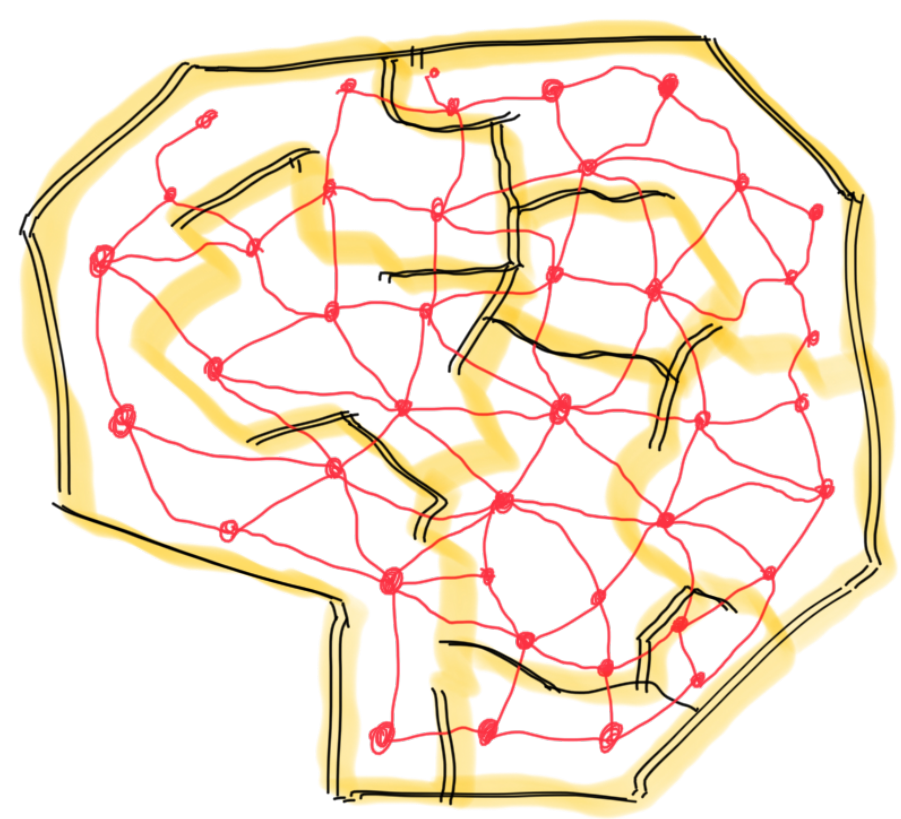}
\caption{The districting plan superimposed on the planar dual, showing
  that it corresponds to a partition of the atoms into connected
  parts; the cost of the solution is the sum of costs of edges of the
  dual that cross between different parts.  In this paper, a
  districting plan is compact to the extent that the total cost of the
  solution is small.}
\end{figure}
\begin{figure}
  \centering
  \includegraphics[scale=.5]{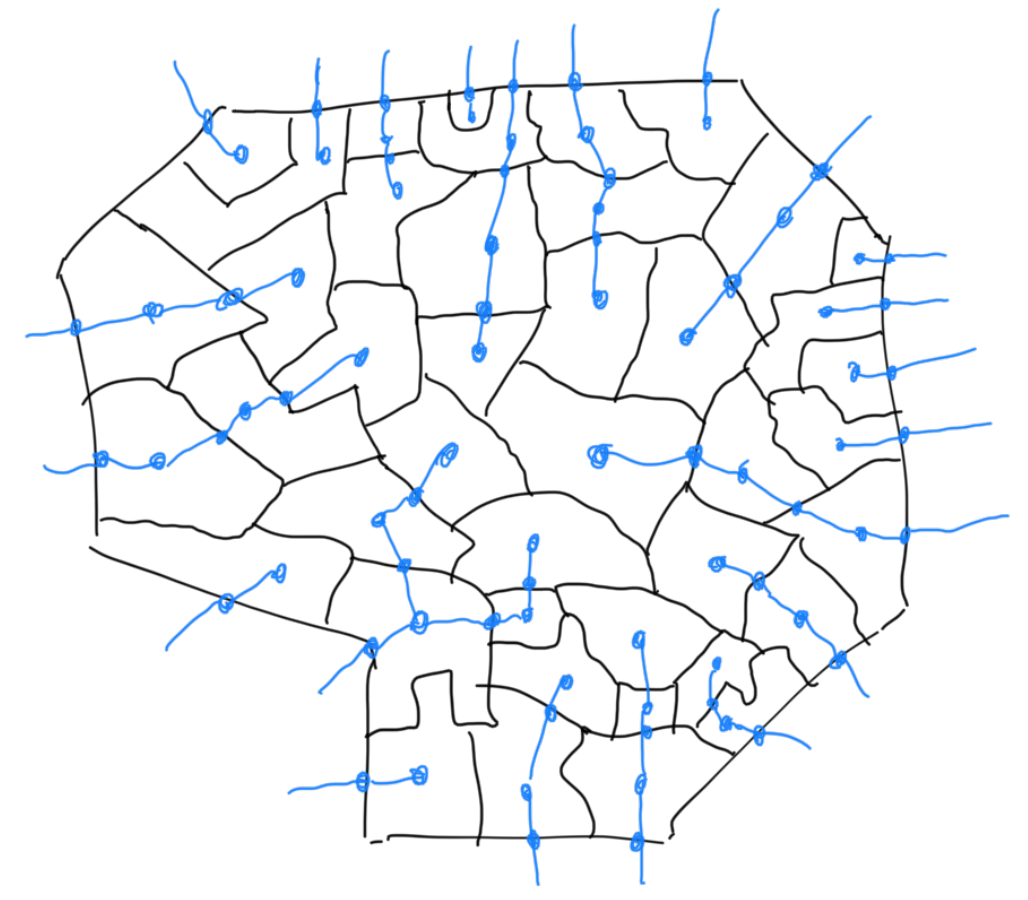}
  \caption{Let $G$ be the graph of atomic regions.  As stated in
    Section~\ref{sec:radial}, the \emph{radial graph} of $G$ has a
    node for every vertex of $G$ and a node for every face of $G$, and
    an edge between a vertex-node and a face-node if the vertex lies
    on the face's boundary.  This diagram shows that every face is
    reachable from the outer face within six hops in the \emph{radial
      graph} of the graph $G$ of atomic regions.  This implies that
    the branchwidth of $G$ and of its dual are at most six.}
  \label{fig:radial}
\end{figure}

\begin{figure} \centering
  \includegraphics[width=3in]{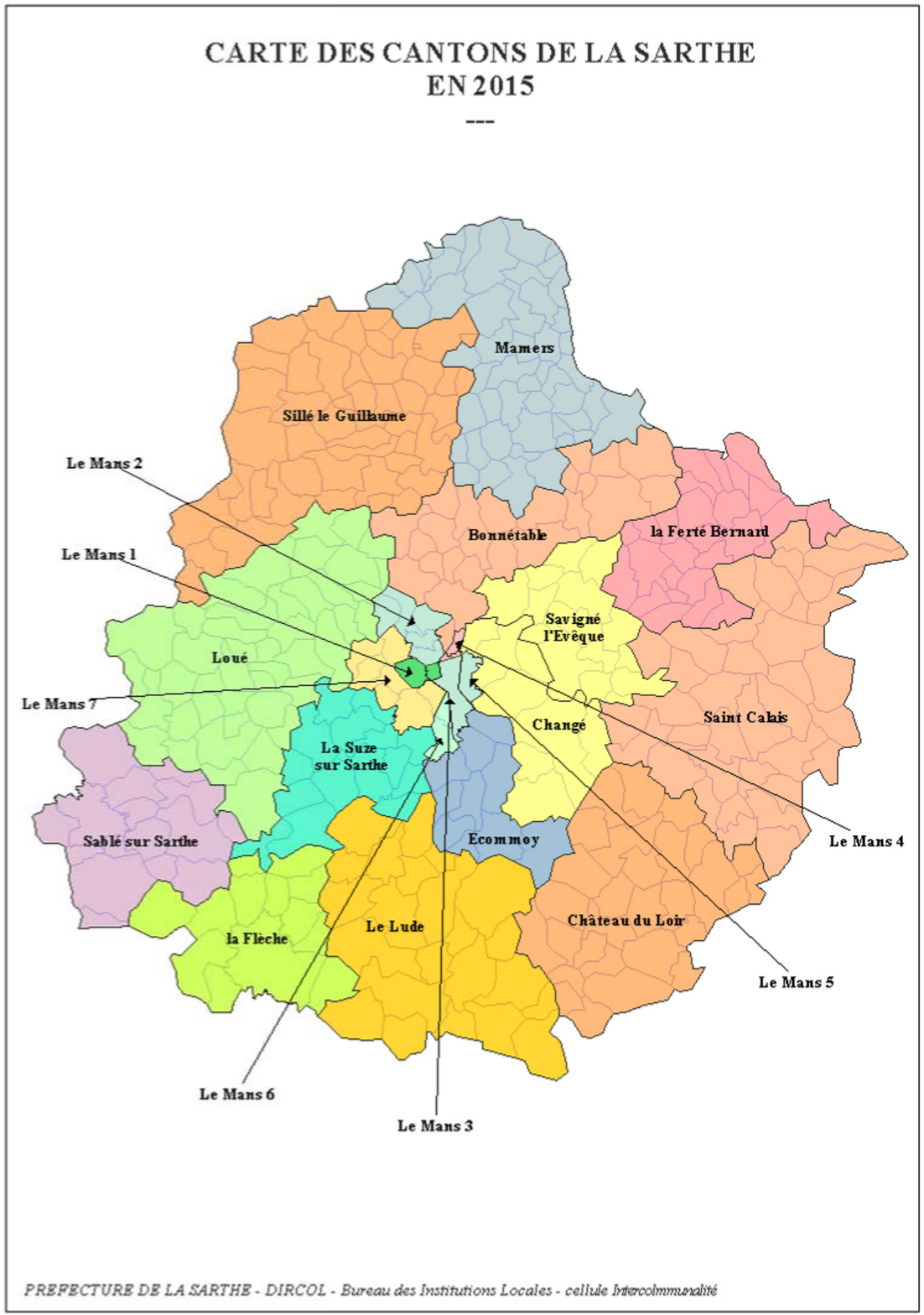}
  \caption{This map shows the twenty-one cantons for the department ``Sarthe'' of
    France.  The cantons are the atomic regions for the redistricting
    of Sarthe.  The corresponding radial graph has radius six, so
    there is a branch decomposition of width $w=6$.  For the upcoming
    redistricting of France, Sarthe must be divided into $k=3$
    districts.}
  \label{fig:France}
\end{figure}

% \caption{This figure shows the construction of the planar dual, and
%   how a districting plan is a partition of the vertices of the dual
%   into connected parts.}
% \label{Duality}

Consider the following computational problems:
\begin{itemize}
  \item \emph{optimization:} Given a planar graph $G$ with vertex weights and
    edge costs, a number $k$, and a weight interval $[L, U)$, find
    the minimum cost of a partition into $k$ connected parts with part-weight in
    $[L,U)$.
    \item \emph{sampling:} Given in addition a number $C$, generate
      uniformly at random a cost-$C$ partition into $k$ connected parts with part-weight
      in $[L,U)$.
\end{itemize}
    
These problem arises in political redistricting.  Each vertex
represents a small geographical
region (such as a \emph{census block} or \emph{census tract} or \emph{county}), and its
weight represents the number of people living in the region.  
Each part is a \emph{district}.  A larger geographic region (such as a
state) must be partitioned into $k$ districts when the state is to be
represented in a legislative body by $k$ people; each district elects
a single representative.  The partition is called a \emph{districting plan}.

The rules governing this partitioning 
vary from place to place, but usually there are (at least) three important goals:
\emph{contiguity}, \emph{compactness}, and \emph{population balance}.\footnote{These terms
are often not formally defined in law.}
\begin{itemize}
  \item \emph{Contiguity} is often interpreted as connectivity; we
    represent this by requiring that the set
    of small regions forming each district is connected via shared
    boundary edges.
    % The problem
    % definition above captures 1 if we take $G$ to be the planar dual
    % of the graph formed by 
\item \emph{Population balance} requires that two different districts
  have approximately equal numbers of people.
  \item 
    One measure of \emph{compactness} is the number
    of pairs of adjacent small regions that lie in distinct
    districts.
  \end{itemize}
Thus In the definitions of 
the \emph{optimization} and \emph{sampling} problems above, the
connectivity constraint reflects the contiguity requirement, the
part-weight constraint reflects the population balance requirement,
and the cost is a measure of compactness.

The \emph{optimization} problem described above arises in
computer-assisted redistricting; an algorithm for solving this problem
could be used to select a districting plan that is optimally compact
subject to contiguity and desired population balance, where
compactness is measured as discussed above.  

The \emph{sampling} problem arises in evaluating a plan; in court
cases~\cite{Chen,Pegden,herschlag2017evaluating,herschlag2018quantifying,pildesbrief} expert witnesses argue that a districting
plan reflects an intention to gerrymander by comparing it to
districting plans randomly sampled from a distribution.  The expert
witnesses use Markov Chain Monte Carlo (MCMC), albeit unfortunately on
Markov chains that have not been shown to be rapidly mixing, which
means that the samples are possibly not chosen according to anything
even close to a uniform distribution.  There have been many papers
addressing random sampling of districting plans (e.g. ~\cite{bangia2017redistricting,carter2019mergesplit,deford2019recombination,herschlag2018quantifying,herschlag2017evaluating}) but, despite
the important role of random sampling in court cases, there are no
results on provably uniform or nearly uniform sampling from a set of
realistic districting plans for a realistic input in a reasonable
amount of time.

It is known that even basic versions of these problems are
NP-hard. If the vertex weights are allowed to very large integers,
expressed in binary, the NP-hardness of \textsc{Subset Sum} already
implies the NP-completeness of partitioning the vertices into two
equal-weight subsets.  However, in application to redistricting the
integers are not very large.  For the purpose of seeking hardness
results, it is better to focus on a special case, the
\emph{unit-weight} case, in which each vertex has weight one.
For this case,
Dyer and Frieze~\cite{DyerF85} showed that, for any fixed
$k\geq 3$, it is NP-hard to find a weight-balanced partition into $k$
connected parts of the vertices of a planar graph.  Najt, Deford, and
Solomon~\cite{NDSolomon} showed that even for $k=2$ and without
the constraint on balance, uniform sampling of partitions into two
connected parts is NP-hard.  

Following Ito et al.~\cite{ItoZN,ItoGZN} and Najt et
al.~\cite{NDSolomon}, we therefore consider a further restriction on
the input graph: we consider graphs with bounded
\emph{branchwidth}/\emph{treewidth}.\footnote{\emph{Treewidth} and
  branchwidth are very similar measures; they are always within a
  small constant factor of each other.  Thus a graph
  has small treewidth if and only if it has small branchwidth.}

The branchwidth of a graph is a measure of how treelike the graph is:
often even an NP-hard graph problem is quickly solvable when the input
is restricted to graphs with low branchwidth.  For planar graphs in
particular, there are known bounds on branchwidth that are relevant to
the application.

\begin{lemma} A planar graph on $n$ vertices has 
  branchwidth $O(\sqrt{n})$.
\end{lemma}

\begin{lemma} A planar graph of diameter $d$ has branchwidth $O(d)$.
\end{lemma}

There is an stronger bound, which we will review in Section~\ref{sec:radial}.

Najt, Deford, and Solomon~\cite{NDSolomon} show that, for any fixed
$k$ and fixed $w$, the optimization and sampling problems
\emph{without the constraint on population balance} can be solved in
polynomial time on graphs of branchwidth at most $w$.\footnote{They
  use treewidth but the results are equivalent.}  Significantly, the
running time is $O(n^c)$ for some constant $c$ independent of $k$ and
$w$.  Such an algorithm is said to be \emph{fixed-parameter tractable}
with respect to $k$ and $w$, meaning that as long as $k$ and $w$ are
fixed, the problem is considered tractable.  Fixed-parameter
tractability is an important and recognized way of coping with
NP-completeness.

However, their result has two disadvantages.  First, as the authors
point out, the big O hides a constant that is astronomical; for NP-hard
problems, it is expected that the dependence on the parameters be
exponential but in this case it is a tower of exponentials.  As the
authors state, the constants in the theorems on which they rely are
``too large to be practically useful.''

Second, because their algorithm cannot handle the constraint on
population balance, the algorithm would not be applicable to
redistricting even if it were tractable.  The authors discuss (Remark
5.11 in~\cite{NDSolomon}) the extension of their approach to handle
balance: ``It is easy to add a relational formula...that restricts our
count to only balanced connected $k$-partitions....  From this it
should follow that ... [the problems are tractable].   However ... the
corresponding meta-theorem appears to be missing from the
literature.''

In our first result, we show that in fact what they seek does not
exist: under a standard complexity-theoretic assumption, \textbf{there is
no algorithm that is fixed-parameter tractable with respect to
both $k$ and $w$.}

More precisely, we use the analogue of NP-hardness for
fixed-parameter tractability, $W[1]$-hardness.

\begin{theorem} \label{thm:hardness}
For unit weights, finding a weight-balanced $k$-partition of a planar graph of width $w$
into connected parts is $W[1]$-hard with respect to $k+w$.
\end{theorem}

In the theory of fixed-parameter tractability (see
e.g. Section 13.4 of~\cite{parameterized-algorithms}) this is strong
evidence that no algorithm exists with a running time of the form
$O(f(k,w) n^c)$ for fixed $c$ independent of $k$ and $w$.

This is bad news but there is a silver lining.  The lower bound guides
us in seeking good algorithms, and it does not rule out an algorithm
that has a running time of the form $f(k) n^{O(w)}$ or
$f(w)n^{O(k)}$.  That is, according to the theory, while there is no
algorithm that is fixed-parameter tractable with respect to both $k$
and $w$ simultaneously, there \emph{could} be one that is
fixed-parameter tractable with respect to $k$ alone and one that is
fixed-parameter tractable with respect to $w$ alone.

These turn out to be true.  Ito et al.~\cite{ItoZN,ItoGZN}
show that, even for general (not necessarily planar) graphs there is
an algorithm with running time $O((w+1)^{2(w+1)} U^{2(w+1)} k^2n)$,
where $U$ is the upper bound on the part weights.  Thus for unit
weights, the running time is $O((w+1)^{2(w+1)} n^{2w+3})$.
We can significantly reduce the exponent $2(w+1)$, and, by restricting
to planar graphs, can also reduce the constant.

However, for the application we have in mind this is not the bound to
try for.  In some real-world redistricting problems, the number $k$ of
districts is very small.\footnote{For example, in the U.S. there are
  seventeen states for which the number of districts is greater than
  one but no more than five.  In one of the most complex redistricting
  problems for France, the number of districts is three.}  The number
of atoms of course tends to be much larger, but the diameter of the
graph is in some cases not so large.  Thus we need an algorithm that
can tolerate a very small number $k$ of districts and a moderately
small branchwidth $w$.

\begin{theorem} \label{thm:algorithms} For the optimization problem
  and the sampling problem, there are algorithms that run in
  $O(c^w u^ksn(\log U +\log S) +n^3)$ time, where $c$ is a constant, $k$ is the
  number of districts, $w\geq k$ is an upper bound on the branchwidth
  of the planar graph, $n$ is the number of vertices of the graph,
  $U$ is the upper bound on the weight of a part, and $S$ is an upper
  bound on the cost of a desired solution.
\end{theorem}

\textbf{Remarks:}
\begin{enumerate}
  \item  In the unit-cost case (every edge cost is one), $S \leq n$. 
\item In the unit-weight, unit-cost case, the running time is $O(c^w n^{k+2}\log n)$.
\item For practical use the input weights need not be the
populations of the atoms; if approximate population is acceptable, the
weight of an atom with population $p$ can be, e.g., $\lceil p/100 \rceil$.
\item The $n^3$ term in the running time accounts for the time required to find
the branchwidth of the input graph; as we will see, that is usually unnecessary in
our application to redistricting.
\end{enumerate}

An implementation of the optimization algorithm has been developed;
another paper will describe the implementation and report on experiments in which the algorithms are applied to
real (but small) redistricting problems.

\section{Preliminaries}

\subsection{Branchwidth} \label{sec:branchwidth}

A \emph{branch decomposition} of a graph $G$ is a rooted binary tree with the
following properties:
\begin{enumerate}
\item Each node $x$ is labeled with a subset $C(x)$ of the edges of $G$.
\item The leaves correspond to the edges of $G$: for each edge $e$,
  there is a leaf $x$ such that $C(x)=\set{e}$.
  \item For each node $x$ with children $x_1$ and $x_2$, $C(x)$ is the
    disjoint union of $C(x_1)$ and $C(x_2)$.
  \end{enumerate}
We refer to a set $C(x)$ as a \emph{branch cluster}.  A vertex $v$ of
$G$ is a \emph{boundary vertex} of $C(x)$ if $G$ has at least one edge
incident to $v$ that is in $C(x)$ and at least one edge incident to
$v$ that is not in $C(x)$.  The \emph{width} of a branch cluster is
the number of boundary vertices, and the width of a branch
decomposition is the maximum cluster width.  The branchwidth of a
graph is the minimum $w$ such that the graph has a branch
decomposition of width $w$.  

For many optimization problems in graphs, if the input graph is
required to have small branchwidth then there is a fast algorithm,
often linear time or nearly linear time, and often this algorithm can
be adapted to do uniform random sampling of solutions.  Therefore Najt, Deford, and
Solomon~\cite{NDSolomon} had good reason to expect that there would be
a polynomial-time algorithm to sample from balanced partitions where
the degree of the polynomial was independent of $w$ and $k$.

\subsection{Radial graph} \label{sec:radial}

For a planar embedded graph $G$, the radial graph of $G$ has a node
 for every vertex of $G$ and a node for
every face of $G$, and an edge between a vertex-node and a
face-node if the vertex lies on the face's boundary.  Note that
the radial graph of $G$ is isomorphic to the radial graph of the
dual of $G$.  There is a linear-time algorithm that, given a
planar embedded graph $G$ and a node $r$ of the radial graph, returns a
branch decomposition whose width is at most the number of hops
required to reach every node of the radial graph from $r$ (see,
e.g.,~\cite{planarity}).  For example, Figure~\ref{fig:radial}
shows that the number of hops required is at most six, so the
linear-time algorithm would return a branch decomposition of width
$w$ at most six.

Using this result, some real-world redistricting graphs can be
shown to have moderately small branchwidth.  For example,
Figure~\ref{fig:France} shows a department of France, Sarthe, that
will need to be divided into $k=3$ districts.  The number of hops
required for this example is six, so we would get a branch
decomposition of width $w$ at most six.
    
\subsection{Sphere-cut decomposition}
\label{sec:sphere-cut}

The branch decomposition of a planar embedded graph can be assumed to
have a special form.  The radial graph of $G$
can be drawn on top of the embedding of $G$ so that a face-node is
embedded in the interior of a face of $G$ and a vertex-node is
embedded in the same location as the corresponding vertex.  We can
assume that the branch decomposition has the property that
corresponding to each branch cluster $C$ is a cycle in the radial
graph that encloses exactly the edges belonging to the cluster $C$,
and the vertices on the boundary of this cluster are the vertex-nodes
on the cycle.  This is called a \emph{sphere-cut
  decomposition}~\cite{DornPBF10}.  If the branch decomposition is
derived from the radial graph using the linear-time algorithm
mentioned above, the sphere-cut decomposition comes for free.
Otherwise, there is an $O(n^3)$ algorithm to find a given planar graph's least-width
branch decomposition, and if this algorithm is used it again gives a
sphere-cut decomposition.

\section{Related work}

There is a vast literature on partitioning graphs, in particular on
partitions that are in a sense balanced.  In particular, in the area
of decomposition of planar graphs, there are algorithms~\cite{Rao87,ParkPhillips,Rao92} for
\emph{sparsest cut} and \emph{quotient cut}, in which the goal is
essentially to break off a single piece
such that the cost of the cut is small compared to the amount of
weight on the smaller side.  The single piece can be required to be connected.
There are approximation algorithms for variants of balanced
partition~\cite{GargSV99,FoxKM15} into two pieces.  These only address
partitioning into $k=2$ pieces, the pieces are not necessarily
connected, and the balance constraint is only approximately satisfied;

There are many papers on algorithms relevant to computer-aided
redistricting (a few examples
are~\cite{helbig1972political,hess,mehrotra_optimization_1998,garfinkel_optimal_1970}),
including papers that have appeared at SIGSPATIAL~\cite{EppsteinGKM17,Cohen-AddadKY18}.

Finally, there many papers on $W[1]$-hardness and more generally
lower bounds on fixed-parameter tractability, as this is a
well-studied area of theoretical computer science.
Our result is somewhat rare in 
that most graph problems are fixed-parameter tractable with respect to
branchwidth/treewidth.  However, there are by now other
$W[1]$-hardness results with respect to
treewidth~\cite{DBLP:conf/iwpec/DomLSV08,DBLP:conf/iwpec/BodlaenderLP09,DBLP:journals/iandc/FellowsFLRSST11,DBLP:conf/approx/MarxSS16,DBLP:journals/siamdm/GutinJW16,DBLP:journals/corr/GuptaSZ17} and a few
results~\cite{DBLP:conf/iwpec/BodlaenderLP09,DBLP:journals/algorithmica/FeldmannM20} were previously known
even under the restriction that the input graph must be planar.

\section{W[1]-Hardness}
In this section, we show that the problem is W[1]-hard parameterized by $k+w$,
where $k$ is the number of districts and $w$ the treewidth of the graph.

We start with the following lemma that shows that it is enough to prove
that a more structure version of the problem (bounded vertex weights, each region
must have size greater than 1) is W[1]-hard.
\begin{lemma}
  \label{lem:unweighted}
  If the planar vertex-weighted version of the \problem~is W[1]-hard parameterized
  by $k+w$ when the total weight of each region should be greater than 1, and
  the smallest weight is 1 and the largest weight is $|V|^{c}$ for some constant $c$,
  then the planar unweighted version of the problem is W[1]-hard parameterized
  by $k+w$.
\end{lemma}
\begin{proof}
  Consider the following transformations of a vertex-weighted 
  instance of the \problem. First, rescale all the weights of vertices
  by a factor $W/w$ where $W$ is the largest vertex weight and $w$ is the
  largest vertex weight.
  For each vertex $v$ of weight $w(v)$, create
  $w(v)-1$ unit-weight \emph{dummy} vertices and connect each of them to $v$
  with a single edge, then remove the weight of $v$.

  This yields a unit-weight graph which satisfies the following properties.
  First, if the input graph was planar, then the resulting
  graph is also planar.
  Second, since the ratio $W/w$ is at most $|V|^c$, the total number of
  vertices in the new graph is at most $|V|^{c+1}$.
  Finally, any solution for the \problem~ on the vertex-weighted graph
  can be associated to a solution for the \problem~ on the unit-weight graph:
  for each vertex $v$ of the original graph, assign each of the $w(v)-1$ dummy
  vertices to the same region as $v$. We have that the associated solution
  has connected regions of exactly the same weight as the solution in
  the weighted graph. Moreover, we claim that any solution for the unit-weight
  graph is associated to a solution of the input weighted graph: this follows from
  the assumption that the prescribed weights for the regions is greater than 1
  and that the regions must be connected. Thus for each vertex $v$,
  in any solution all the $w(v)-1$ dummy vertices must belong to the region of $v$.

  Therefore, if the planar vertex-weighted version of the \problem~ is W[1]-hard parameterized
  by $k+w$ when the smallest weight is at least 1, the total weight of each region should be greater
  than 1, and the total weight of the graph is at most $|V|^{c}$ for some constant $c$,
  then the planar unit-weight version of the problem is W[1]-hard parameterized
  by $k+w$.
\end{proof}

By Lemma~\ref{lem:unweighted}, we can focus without loss of generality
on instances where the vertices have weights between 1 and $|V|^c$ for some fixed
constant $c$. We next show that the problem is W[1]-hard on these instances.

We reduce from the Bin Packing problem with polynomial weights.
Given a set of integer values $v_1,\ldots,v_n$ and two integers $B$ and $k$,
the \emph{Bin Packing} problem asks to decide whether there exists a
partition of $v_1,\ldots,v_n$ into
$k$ parts such that for each part of the partition, the sum of the values is at most
$B$.
The Bin Packing problem with polynomially bounded weights assumes that
there exists a constant $c$ such that $B = O(n^{c})$. Note that for the case
where the weights are polynomially bounded, we can assume w.l.o.g. that the sum
of the weights is exactly $kB$ by adding $kB - \sum_{i=1}^n v_i$ elements of value 1.
Since the weights are polynomially bounded and that each weight is integer we have
that (1) the total number of new elements added is polynomial in $n$, hence the size
of the problem is polynomial in $n$, and (2) there is
a solution to the original problem if and only if there is a solution to the new problem:
the new elements can be added to fill up the bins that are not full in the solution of the
original problem.

We will make use of the following theorem of Jansen et al.~\cite{JansenKMS13}.
\begin{theorem}[\cite{JansenKMS13}]
  The Bin Packing problem with polynomial weights is W[1]-hard parameterized
  by the number of bins $k$.
  Moreover, there is no $f(k)n^{o(k/log k)}$ time algorithm assuming the
  exponential time hypothesis (ETH).
\end{theorem}

We now proceed to the proof of Theorem~\ref{thm:hardness}.
From an instance of Bin Packing with polynomially bounded weights and whose
sum of weights is $kB$, create
the following instance for the \problem.
For each $i \in [2k+1]$, create 
$$\ell_i =\begin{cases} k &\mbox{if $i$ is odd}\\
k+1 &\mbox{if $i$ is even}
\end{cases}
$$
vertices $s_i^1,\ldots,s_i^{\ell_i}$. Let $S_i = \{s_i^1,\ldots,s^{\ell(i)}_i\}$.
Moreover, for each odd $i < n$, for each $1 \le j \le k$, connect
$s^j_i$ to $s^j_{i-1}$ and $s^j_{i+1}$, and when $j < k$, also to $s^{j+1}_{i-1}$ and $s^{j+1}_{i+1}$.
Let $G$ be the resulting graph.

It is easy to see that $G$ is planar. We let 
$f_{\infty}$ be the longest face:\\
$\{s^1_1,\ldots,s^k_1, s^{k+1}_2, s^k_3, \ldots,s^{k}_{2n+1}, s^{k-1}_{2n+1}, \ldots, s^{1}_{2n+1}, s^{1}_{2n},\ldots,s^{1}_2\}$.\\
We claim that the treewidth of the graph is at most $7k$. To show this
we argue that the face-vertex incidence graph $\bar{G}$ of $G$ has
diameter at most $2k+4$ and by Lemma~\ref{lem:unweighted} this immediately yields
that the treewidth of $G$ is at most $10k$. We show that each vertex of $\bar{G}$
is at hop-distance at most $k+2$ of the vertex corresponding to $f_{\infty}$.
Indeed, consider a vertex $s^j_i$ (for a face, consider a vertex $s^j_i$ on that face).
Recall that for each $i_0,j_0$, we have that $s^{j_0}_{i_0}$ is adjacent to $s^{j_0}_{i+1}$
and $s^{j_0+1}_{i+1}$ and so, 
$s^j_i$ is at hop-distance at most $k+1$ from either $s^{\ell(i)}_i$ or $s^1_i$
in $\bar{G}$. Moreover both $s^1_i$ and $s^{\ell(i)}_n$ are on face $f_{\infty}$ and so
$s^j_i$ is at hop-distance at most $k+2$ from $f_{\infty}$ in $\bar{G}$.
Hence the treewidth of $G$ is at most $10k$.

Our next step is to assign weights to the vertices.
Then, we set the weight $w(s^j_i)$ of every vertex $s^j_i$ of $\{s^1_1,\ldots,s^k_1\}$ to be
$(kB)^2$ and the weight $w(s^j_i)$ of every vertex $s^j_i$ of $\{s^1_{2n+1},\ldots,s^{k}_{2n+1}\}$
to be $(kB)^4$. For each odd $i \neq 1,2n+1$ we set a weight of $1/(2n-2)$.
Finally, we set the weight of each vertex $s^{j}_i$ where $i$ is even to be
$v_i$.
Let $T = (kB)^2 + (kB)^4 + 1/2 + kB$, and recall that $kB = \sum_{i=1}^n v_i$.

\begin{fact}
  \label{fact:weights}
  Consider a set $S$ of vertices containing exactly one vertex of $S_i$ for each $i$.
  Then the sum of the weights of the vertices in $S$ is $T$.
\end{fact}

We now make the target weight of each region to be $(kB)^2 + (kB)^4 + kB + B = T+B$,
We have the following lemma.
\begin{lemma}
  \label{lem:firstlast}
  In any feasible solution to the \problem, there is exactly 1 vertex of
  $\{s^1_1,\ldots,s^k_1\}$ and exactly 1 vertex of
  $\{s^1_n,\ldots,s^{\ell(n)}_n\}$ in each region.
\end{lemma}
\begin{proof}
  Recall that by definition we have that
  $\sum_{i=1}^n v_i = k B$. Moreover, the number
  of vertices with weight equal to $(kB)^2$ is exactly $k$.
  Thus, since the target weight of each region is $(kB)^2 + (kB)^4 + B + kB$,
  each region has to contain exactly 1 vertex from 
  $\{s^1_1,\ldots,s^k_1\}$
  and exactly 1 vertex from
  $\{s^1_n,\ldots,s^{\ell(n)}_n\}$.
\end{proof}

%% Lemma~\ref{lem:firstlast} implies that in a feasible solution, each region
%% contains a path from a vertex of $\{s^1_1,\ldots,s^k_1\}$ to a vertex of
%% $\{s^1_n,\ldots,s^{\ell(n)}_n\}$. Each region thus contains a vertex
%% of $\{s^1_i,\ldots,s^{\ell(n)}_{i}\}$ for each $i$.
We now turn to the proof of completeness and soundness of the reduction.
We first show that if there exists a solution to the Bin Packing instance, namely that
there is a partition into $k$ parts such that for each part of the partition,
the sum of the values is $B$, then there exists a feasible solution
to the problem.
Indeed, consider a solution to the Bin Packing instance $\{B_1,\ldots,B_k\}$
and construct the following solution to the \problem.
For each odd $i$, assign vertices $s^1_i,\ldots,s^k_i$ to regions $R_1,\ldots,R_k$ respectively.
For each $i \in [n]$, perform the following assignment for the even rows. Let $u_i$ be the integer in $[k]$
such that $v_{i} \in B_{u_i}$.
Assign all vertices $s^{1}_{2i},\ldots,s^{u_i-1}_{2i}$ to regions $R_1,\ldots R_{u_i-1}$ respectively. Assign
both vertices $s^{u_i}_{2i}$ and $s^{u_i+1}_{2i}$ to region $R_{u_i}$. Assign
all vertices $s^{u_i+2}_{2i}, \ldots s^{k+1}_{2i}$ to regions $R_{u_i+1},\ldots R_k$.
The connectivity of the regions follows from the fact that for each odd $i$,
$s^{j}_i$ is connected to both $s^{j}_{i+1}$ and $s^{j+1}_{i+1}$ and to
both $s^{j}_{i-1}$ and $s^{j+1}_{i-1}$.

We then bound the total weight of each region. Let's partition the
vertices of a region $R_j$ into two: Let $S_{R_j}$ be a set
that contains one vertex from each  $S_i$ and
let $\bar{S_{R_j}}$ be the rest of the elements.
The total weight of the vertices in $S_{R_j}$ is by Fact~\ref{fact:weights}
exactly $T$.
The total weight of the remaining vertices corresponds to the sum of
the values $v_i$ such that $|R_j \cap S_i| = 2$ which is
$\sum_{v_i \in B_j} v_i = B$ since it is a solution to the Bin Packing problem.
Hence the total weight of the region is 
$T + B$, as prescribed by the \problem.

We finally prove that if there exists a solution for the \problem~with the prescribed region weights,
then there exists a solution to the Bin Packing problem. Let $R_1,\ldots,R_k$ be the solution
to the problem.
By Lemma~\ref{lem:firstlast}, each region contains one vertex of $s^{1}_1, \ldots s^{k}_1$ and one vertex
of $s^{1}_1, \ldots s^{k}_{2n+1}$. Since the regions are required to be connected, there exists
a path joining these two vertices and so by the pigeonhole principle
for each odd $i$, each region contains exactly
one vertex of $s^{1}_i, \ldots s^{k}_{i}$. Moreover for each even $i$, each region contains at
least one vertex of $s^{1}_i, \ldots s^{k+1}_{i}$ and exactly one region contains two vertices.
Let $\phi(i) \in [k]$ be such that $|R_{\phi(i)} \cap \{s^{1}_i, \ldots s^{k+1}_{i}\}|  = 2$.
We now define the following solution for the Bin Packing problem.
Define the $j$th bin as $B_j = \{v_i \mid \phi(i) = j\}$. We claim that for each bin $B_j$
the sum of the weights of the elements in $B_j$ is exactly $B$.
Indeed, observe that region $R_j$ contains exactly one vertex of $s^{1}_i, \ldots s^{k}_{i}$
for each odd $i$ and exactly one vertex of $s^{1}_i, \ldots s^{k+1}_{i}$ for each even $i$ except
for the sets $s^{1}_i, \ldots s^{k+1}_{i}$ where $\phi(i) = j$ for which it contains two vertices.
Thus by Fact~\ref{fact:weights}, the total sum of the weights is $T + \sum_{i \mid \phi(i) = j} v_i$ and
since the target weight is $T+B$ we have that $\sum_{i \mid \phi(i) = j} v_i = B$.
Since the weight of $B_j$ is exactly $\sum_{i \mid \phi(i) = j} v_i$ the proof is complete.

\section{Algorithm} \label{sec:alg}

In this section, we describe the algorithms of
Theorem~\ref{thm:algorithms}.  In describing the algorithm, we will
focus on simplicity rather than on achieving the best constant
possible as the base of $k$.  

\subsection{Partitions}

A \emph{partition} of a finite set $\Omega$ is a collection of
disjoint subsets of $\Omega$ whose union is $\Omega$.  A partition
defines an equivalence relation on $\Omega$: two elements are
equivalent if they are in the same subset.

There is a partial order on partitions of $\Omega$: $\pi_1 \prec
\pi_2$ if every part of $\pi_1$ is a subset of a part of $\pi_2$.
This partial order is a lattice.  In particular, for any pair
$\pi_1,\pi_2$ of partitions of $\Omega$, there is a unique minimal
partition $\pi_3$ such that $\pi_1 \prec \pi_3$ and $\pi_2 \prec
\pi_3$.  (By \emph{minimal}, we mean that for any partition $\pi_4$
such that $\pi_1 \prec \pi_4$ and $\pi_2 \prec \pi_4$, it is the case
that $\pi_3 \prec \pi_4$.)  This unique minimal partition is called
the \emph{join} of $\pi_1$ and $\pi_2$, and is denoted $\pi_1 \vee
\pi_2$.

It is easy to compute $\pi_1 \vee \pi_2$: initialize $\pi := \pi_1$,
and then repeatedly merge parts that intersect a common part of
$\pi_2$.

In a slight abuse of notation, we define the join of a partition
$\pi_1$ of one finite set $\Omega_1$ and a partition $\pi_2$ of
another finite set $\Omega_2$.  The result, again written $\pi_1 \vee
\pi_2$, is a partition of $\Omega_1 \cup \Omega_2$.  It can be defined
algorithmically: iniitalize $\pi$ to consist of the parts of $\pi_2$,
together with a singleton part $\set{\omega}$ for each $\omega\in
\Omega_2 - \Omega_1$.  Then repeatedly merge parts of $\pi$ that
intersect a common part of $\pi_2$.

\subsection{Noncrossing partitions}
\label{sec:partitions}

The sphere-cut decomposition is algorithmically useful because it
restricts the way a graph-theoretic structure (such as a solution) can
interact with each cluster.  For a cluster $C$, consider the
corresponding cycle in the radial graph, and let $\theta_C$ be the
cyclic permutation $(v_1\ v_2\ \cdots\ v_m)$ of boundary vertices in
the order in which they appear in the radial cycle.  (By a slight
abuse of notation, we may also interpret $\theta_C$ as the set
$\set{v_1, \ldots, v_m}$.

First consider a partition $\rho^\ins$
of the vertices incident to edges belonging to $C$, with the property
that each part induces a connected subgraph of $C$.  Planarity
implies that the partition induced by $\rho^\ins$ on the boundary vertices
$\set{v_1, \ldots, v_m}$ has a special property.

\begin{definition}
Let $\pi$ be a partition of the set $\set{1,\ldots, w}$.  We say $\pi$ is
\emph{crossing} if there are integers $a<b<c<d$ such that one part
contains $a$ and $c$ and another part contains $b$ and $d$.
\end{definition}

It follows from connectivity that the partition induced by $\rho^\ins$
on the boundary vertices $\theta_C$ is a noncrossing partition.
Similarly, let $\rho^\out$ be a partition of the vertices incident to
edges that do \emph{not} belong to $C$; then $\rho^\out$ induces a
noncrossing partition on the boundary vertices of $C$.

The asymptotics of the Catalan numbers imply the following (see,
e.g.,~\cite{DornPBF10}).
\begin{lemma} \label{lem:noncrossing-partitions} There is a constant
  $c_1$ such that the number of noncrossing partitions of $\set{1,
    \ldots, w}$ is $O(c_1^w)$.
\end{lemma}

Finally, suppose $\rho$ is a partition of all vertices of $G$ such
that each part is connected.  Then $\rho=\rho^\ins \vee \rho^\out$
where $\rho^\ins$ is a 
partition of the vertices incident to edges in $C$ (in
which each part is connected) and $\rho^\out$ is a partition  of the
vertices incident to edges not in $C$ (in which each part is
connected). 

Because the only vertices in both $\rho^\ins$ and $\rho^\out$ are
those in $\theta_C$, the partition $\rho$ induces on $\theta_C$ is
$\pi^\ins \vee \pi^\out$ where $\pi^\ins$ is the partition induced on
$\theta_C$ by $\rho^\ins$ and $\pi^\out$ is the
partition induced on $\theta_C$ by $\rho^\out$.

\subsection{Algorithm overview}

\newcommand{\w}{\boldsymbol{w}}

The algorithms for optimization and
sampling are closely related.  
% For now we will describe the
% optimization algorithm.

The algorithms are based on dynamic programming using the sphere-cut
decomposition of the planar embedded input graph $G$. 

Each algorithm considers every vertex $v$ of the input graph and selects
one edge $e$ that is incident to $v$, and designates each branch cluster that
contains $e$ as a \emph{home cluster} for $v$.

We define a \emph{topological configuration} of a cluster $C$ to be a pair
$(\pi^\ins, \pi^\out)$ of noncrossing partitions of $\theta_C$ with
the following property:
\begin{equation} \label{eqn:at-most-k-parts}
  \text{$\pi^\ins \vee \pi^\out$ has at most $k$ parts.}
\end{equation}
The intended interpretation is that there exist $\rho^\ins$ and
$\rho^\out$ as defined in Section~\ref{sec:partitions} such that
$\phi^\ins$ is the partition $\rho^\ins$ induces on $\theta_C$ and
$\phi^\out$ is the partition $\rho^\out$ induces on $\theta_C$.

\newcommand{\representatives}{\text{representatives}}

We can assume that the vertices of the graph are assigned unique
integer IDs, and that therefore there is a fixed total ordering of
$\theta_C$.  Based on this total ordering, for any partition $\pi$ of
$\theta_C$, let $p$ be the number of parts of $\pi$, and define
$\representatives(\pi)$ to be the $p$-vector 
$(v_1, v_2, \ldots, v_p)$ obtained as follows:
\begin{itemize}
  \item $v_1$ is the
smallest-ID vertex in
$\theta_C$,
\item $v_2$ is the smallest-ID vertex in $\theta_C$ that is not
  in the same part as $v_1$,
  \item $v_2$ is the smallest-ID vertex in $\theta_C$ that is not
  in the same part as $v_1$ and is not in the same part as $v_2$,
\end{itemize}
and so on.

This induces a fixed total ordering of the parts of
$\pi^\ins \vee \pi^\out$.

We define a \emph{weight configuration} of $C$ to be a $k$-vector
$\w = (w_1, \ldots, w_k)$ where each $w_i$ is a nonnegative integer
less than $U$.  There are $U^k$ such vectors.

We define a \emph{weight/cost configuration} of $C$ to be a $k$-vector
together with a nonnegative integer $s$ less than $S$.  There are $U^kS$
such configurations.

We define a \emph{configuration} of $C$ to be a pair consisting of a
topological configuration and a weight/cost configuration.   The number of
configurations of $C$ is bounded by $c^w U^kS$.

The algorithms use dynamic programming to construct, for each cluster
$C$, a table $T_C$ indexed by configurations of $C$.  In the case of
optimization, the table entry $T_C[\Psi]$ corresponding to a
configuration $\Psi$ is \emph{true} or \emph{false}.  For sampling, $T_C[\Psi]$ is a cardinality.

\newcommand{\countx}{\text{count}}

Let $\Psi = ((\pi^\ins,\pi^\out), ((w_1, \ldots, w_k), s))$ be a
configuration of $C$.
Let $\countx(\Psi)$ be the number of partitions $\rho^\ins$ of the vertices incident to
    edges belonging to $C$ with the following properties:
\begin{itemize}
  \item $\rho^\ins$ induces $\pi^\ins$ on $\theta_C$.
\item  Let $\pi=\pi^\ins \vee \phi^\out$.  Let $\representatives(\pi)
  = (v_1, \ldots, v_p)$. 
  Then for $j=1, \ldots, p$, $w_j$ is the total weight of vertices $v$
  for which $C$ is a home cluster and such that $v$ belongs to the
  same part of $\rho^\ins \vee \pi^\out$ as $v_j$.
\end{itemize}

For optimization, $T_C[\Psi]$ is true if $\countx(\Psi)$ is nonzero.
For sampling, $T_C[\Psi]=\countx(\Psi)$.
We describe in Section~\ref{sec:DP} how to populate these tables.
Next we describe how they can be used to solve the problems.

\subsection{Using the tables}

For the root cluster $\hat C$, the cluster that contains all edges of
$G$, $\theta_{\hat C}$ is empty.  Therefore there is only one
partition of $\theta_{\hat C}$, the trivial partition $\pi_0$ consisting of a
single part, the empty set.

To detemine the optimum cost in the optimization problem, simply find
the minimum nonnegative integer $s$ such that, for some
$\w = (w_1, \ldots, w_k)$ such that each $w_i$ lies in $[L,U)$, the
entry $T_{\hat C}[((\pi_0, \pi_0), (\w, s))]$ is \emph{true}.  To find
the solution with this cost, the algorithm needs to find a
``corresponding'' configuration for each leaf cluster $C(\set{uv})$ ; that
configuration tells the algorithm whether the two endpoints $u$ and
$v$ are in the same district.  This information is obtained by a
recursive algorithm, which we presently describe.

Let $C_0$ be a cluster with child clusters $C_1$ and $C_2$.  For
$i=0,1,2$, let $(\pi^\ins_i,\pi^\out_i)$ be a topological configuration for cluster $C_i$.  Then we
say these topological configurations are \emph{consistent} if the following properties
hold:
\begin{itemize}
\item For $i=1,2$, $\pi^\out_i = \pi^\out_0 \vee \pi^\ins_{3-i}$.
\item $\pi^\ins_0 = \pi^\ins_1 \vee \pi^\ins_2$.
\end{itemize}
For $i=0,1,2$, let $(\w_i,s_i)$ be a weight/cost configuration for
$C_i$.  We say they are consistent if  $\w_0 = \w_1 + \w_2$ and $s_0 = s_1+s_2$.

Finally, for
$i=0,1,2$, let $\Psi_i = ((\pi^\ins_i,\pi^\out_i),(\w_i,s_i))$ be a configuration for cluster $C_i$.  Then we
say $\Psi_1, \Psi_2, \Psi_3$ are consistent if the topological configurations are
consistent and the weight/cost configurations are consistent.

\begin{lemma} \label{lem:count}  For a configuration $\Psi_0$ of $C_0$,
  $\countx(\Psi_0) = \sum_{\Psi_1,\Psi_2}\countx(\Psi_1) \cdot
  \countx(\Psi_2)$ where the sum is over pairs $(\Psi_1, \Psi_2)$ of
  configurations of $C_1,C_2$ such that $\Psi_0,\Psi_1,\Psi_2$ are consistent.
\end{lemma}

The recursive algorithm, given a configuration $\Psi$ for a cluster
$C$ such that $T_C[\Psi]$ is \emph{true}, finds configurations for all
the clusters that are descendants of $C$ such that, for each nonleaf
descendant and its children, the corresponding configurations are
consistent; for each descendant cluster $C'$, the configuration
$\Psi'$ selected for it must have the property that $T_{C'}[\Psi']$ is
\emph{true}.

The algorithm is straightforward:
\begin{tabbing}
  define \textsc{Descend}$(C_0, \Psi_0)$:\\
\  \emph{precondition:} $T_{C_0}[\Psi_0]=\text{\it true}$\\
  assign $\Psi_0$ to $C_0$\\
  \ if $C_0$ is not a leaf config\\
  \quad \= for each config $\Psi_1 = ((\pi^\ins_1,\pi^\out_1),
  (\w_1, s_1))$ of $C_0$'s left child $C_1$,\\
  \> \quad \= if $T_{C_1}[\Psi_1]$ is \emph{true}\\
  \> \> \quad \= for each topological config
  $(\pi^\ins_2,\pi^\out_2)$ of $C_0$'s right child $C_2$\\
\> \> \> \quad \= let $(\w_2,s_2)$ be the weight/cost config
of $C_2$ such that\\
\> \>\> \> \quad \=  $\Psi_0, \Psi_1, \Psi_2$ are consistent\\
\> \> \> \> \> \quad \=
where $\Psi_2= ((\pi^\ins_2,\pi^\out_2), (\w_2,s_2))$\\
\> \> \> \> if $T_{C_2}[\Psi_2]$ = \emph{true}\\
\> \> \> \> \> call \textsc{Descend}$(C_1, \Psi_1)$ and \textsc{Descend}$(C_2, \Psi_2)$\\
\> \> \> \> \> return, exiting out of loops
\end{tabbing}

Lemma~\ref{lem:count} shows via induction from root to leaves that the procedure will
successfully find configurations for all clusters that are descendants
of $C_0$.  For the root cluster $\hat C$ and a configuration $\hat \Psi$ of
$\hat C$ such that $T_{\hat C}[\hat \Psi]$ is \emph{true}, consider
the $\Psi_C$ configurations found for each leaf cluster, and let
$(\pi^\ins_C, \pi^\out_C)$ be the topological configuration of $\Psi_C$
Consider the partition
$$\rho = \bigvee_C \pi^\ins_C$$
where the join is over all leaf clusters $C$.   Because there are no
vertices of degree one, for each leaf cluster $C(\set{uv})$, both $u$
and $v$ are boundary vertices, so $\rho$ is a partition of all
vertices of the input graph.  Induction from leaves to root shows that
this partition agrees with the weight/cost part $(\hat \w,\hat s)$ of the configuration
$\hat \Psi$.  In particular, the weights of the parts of $\rho$
correspond to the weights of $\hat w$, and the cost of the partition
equals $\hat s$.

In the step of \textsc{Descend} that selects $(\w_2,s_2)$, there is exactly
one weight/cost config that is consistent (it can be obtained by
permuting the elements of $\w_1$ and then subtracting from $\w_0$ and
subtracting $s_1$ from $s_0$).  By an appropriate choice of an
indexing data structure to represent the tables, we can ensure that the running time of
\textsc{Descend} is within the running time stated in
Theorem~\ref{thm:algorithms}.  For optimization, it remains to show
how to populate the tables.

% For sampling, we use a similar recursive algorithm.  This one, given
% a nonnegative integer $p$, outputs the $p^{th}$ solution.

\begin{tabbing}
  define \textsc{Descend}$(C_0, \Psi_0, p)$:\\
\ \emph{precondition:} $p \leq T_{C_0}[\Psi_0]$\\
\  assign $\Psi_0$ to $C_0$\\
\ if $C_0$ is not a leaf config\\
    \quad \= for each config $\Psi_1 = ((\pi^\ins_1,\pi^\out_1),
(\w_1, s_1))$ of $C_0$'s left child $C_1$,\\
%  \> \quad \= if $T_{C_1}[\Psi_1]>0$\\
  \> \quad \= for each topological config
  $(\pi^\ins_2,\pi^\out_2)$ of $C_0$'s right child $C_2$\\
\> \> \quad \= let $(\w_2,s_2)$ be the weight/cost config
of $C_2$ such that\\
\> \> \> \quad \=  $\Psi_0, \Psi_1, \Psi_2$ are consistent\\
\> \> \> \> \quad \=
where $\Psi_2= ((\pi^\ins_2,\pi^\out_2), (\w_2,s_2))$\\
\> \>     \> $\Delta := T_{C_1}[\Psi_1] \cdot T_ {C_2}[\Psi_2]$\\
    \> \> \> if $p \leq \Delta$\\
    \> \> \> \quad \= $q := \lfloor p / T_ {C_2}[\Psi_2]\rfloor$\\
        \> \> \> \> $r := r \bmod T_ {C_2}[\Psi_2]$\\
\> \> \> \> call \textsc{Descend}$(C_1, \Psi_1, q)$ and \textsc{Descend}$(C_2, \Psi_2, r)$\\
\> \> \> \> return\\
\> \>\> else $p := p - \Delta$ and continue
\end{tabbing}

Induction shows that this procedure, applied to root
cluster $\hat C$ and a configuration $\hat \Psi$ and an integer $p\leq
T_{\hat C}[\hat \Psi]$, selects the $p^{th}$ solution among those
``compatible'' with $\hat \Psi$.  This can be used for random
generation of solutions with given district populations and a given cost.
Again, the running time for the procedure is within that
stated  in Theorem~\ref{thm:algorithms}.

\subsection{Populating the tables}
\label{sec:DP}

For this section, let us focus on the tables needed for sampling.
Populating the table for a leaf cluster is straightforward.
Therefore, suppose $C_0$ is a cluster with children $C_1$ and $C_2$.  
We first observe that, given noncrossing partitions $\pi^\out_0$ of
$\theta_{C_0}$, $\pi^\ins_1$ of $\theta_{C_1}$, and $\pi^\ins_2$ of
$\theta_{C_2}$, there are unique partitions $\pi^\ins_0, \pi^\out_1,
\pi^\out_2$ such that the topological configurations $(\pi^\ins_0, \pi^\out_0),
(\pi^\ins_1, \pi^\out_1), (\pi^\ins_2,\pi^\out_2)$ are consistent.
(The formulas that show this are in the pseucode below.)

For the second observation,  consider a
configuration $\Psi_0 =(\kappa_0,(\w_0,s_0))$ of $C_0$.  
Then $\countx(\Psi_0)$ is
\begin{equation}
 \sum_{\kappa_1,\kappa_2} \sum_{(\w_1,s_1),(\w_2,s_2)}
 \countx((\kappa_1, (\w_1,s_1))) \cdot \countx((\kappa_2, (\w_2,s_2)))
\end{equation}
 where the first sum is over pairs of topological configurations
 $\kappa_1$ for $C_1$ and and $\kappa_2$ where
 $\kappa_0,\kappa_1,\kappa_2$ are consistent, and the second sum is
 over pairs of weight/cost configurations that are consistent with $(\w_0,s_0)$.
Note that because of how weight/cost configuration consistency is
defined, the second sum mimics multivariate polynomial multiplication.
We use these observations to define the procedure that populates the
table for $C_0$ from the tables for $C_1$ and $C_2$.

\begin{tabbing}
  def \textsc{Combine}$(C_0, C_1, C_2)$:\\
 \ initialize each entry of $T_{C_0}$ to zero\\
\ for each noncrossing partition $\pi^\out_0$ of $\theta_{C_0}$\\
\quad \= for each noncrossing partition $\pi^\ins_1$ of $\theta_{C_1}$\\
\> \quad \= for each noncrossing partition $\pi^\ins_2$ of $\theta_{C_2}$\\
\> \> \quad \= $\pi^\out_1 = \pi^\out_0 \vee \pi^\ins_2$\\
\> \> \>           $\pi^\out_2 = \pi^\out_0 \vee \pi^\ins_1$\\
\> \> \>           $\pi^\ins_0 = \pi^\ins_1 \vee \pi^\ins_2$\\
\> \> \>           \emph{comment: } \= now we populate entries of
$T_{C_0}[\cdot]$ indexed by\\
\> \> \>                         \>configurations of $C_0$ with\\
\> \> \>                         \>topological configuration $(\pi^\ins_0,\pi^\out_0)$.\\
\>\>\> for $i=1,2$,\\
\>\>\>\quad \= let $p_i(\xvec,y)$ be a polynomial over variables $x_1,
\ldots, x_k, y$\\
\>\>\>\>\quad\= such that the coefficient of $x_1^{w_1} \cdots x_k^{w_k}y^s$\\
\>\>\>\> \> is $T_{C_i}[((\pi^\ins_0,\pi^\out_0), ((w_1, \ldots,
w_k),s))]$\\
\>\>\> let $p(\xvec,y)$ be the product of $p_1(\xvec,y)$ and $p_2(\xvec,y)$\\
\>\>\> for every weight/cost configuration $((w_1, \ldots, w_k),s)$\\
\>\>\>\> add to $T[((\pi^\ins_0,\pi^\out_0),((w_1, \ldots, w_k),s))]$ the\\
\>\>\>\>\quad \= coefficient of $x_1^{w_1}\cdots x_k^{w_k}y^s$ in $p(\xvec,y)$
\end{tabbing}

The three loops involve at most $c^w$ iterations, for some
constant $c$.  Multivariate polynomial multiplication can be done using
multidimensional FFT.  The time required is $O(N \log N)$, where
$N=U^kS$.  (This use of FFT to speed up an algorithm is by now a
standard algorithmic technique.) It follows that the running time of
the algorithm to populate the tables is as described in
Theorem~\ref{thm:algorithms}.

\bibliography{refs}
\bibliographystyle{plain}

\end{document}